\documentclass[conference,10pt]{IEEEtran}
\usepackage{cite}
\usepackage[dvips]{graphicx}
\usepackage{color}
\usepackage{amstext}
\usepackage{hyperref}
\usepackage{amsfonts}
\usepackage{amsmath,amssymb}
\usepackage{psfrag}
\usepackage{fancyhdr}

\usepackage{mathtools}
\usepackage{amsthm}

\usepackage{subfigure}
\usepackage{epsfig}
\usepackage{csquotes}
\MakeOuterQuote{"}

\usepackage[linesnumbered,ruled,vlined]{algorithm2e}

\theoremstyle{definition}
\newtheorem{defn}{Definition}

\newtheorem{lem}{Lemma}

\usepackage{enumerate}
\usepackage{bbm}
\usepackage{tabularx}
\usepackage{etoolbox}

\let\bbordermatrix\bordermatrix
\patchcmd{\bbordermatrix}{8.75}{4.75}{}{}
\patchcmd{\bbordermatrix}{\left(}{\left[}{}{}
\patchcmd{\bbordermatrix}{\right)}{\right]}{}{}

\begin{document}

\title{Dynamic Proximity-aware Resource Allocation in Vehicle-to-Vehicle (V2V) Communications}
\author{
\IEEEauthorblockN{Muhammad Ikram Ashraf\IEEEauthorrefmark{1},  Mehdi Bennis\IEEEauthorrefmark{1}, Cristina Perfecto\IEEEauthorrefmark{2}, and Walid Saad\IEEEauthorrefmark{3} \\}
\IEEEauthorblockA{\small\IEEEauthorrefmark{1}Centre for Wireless Communications, University of Oulu, Finland, Emails: \{ikram,bennis\}@ee.oulu.fi \\
\IEEEauthorrefmark{2}University of the Basque Country - UPV/EHU, Spain, Email: {cristina.perfecto@ehu.eus} \\
\IEEEauthorrefmark{3}Wireless@VT, Bradley Department of Electrical and Computer Engineering, Virginia Tech, Blacksburg, USA, Email: walids@vt.edu}
}
\maketitle
\begin{abstract}
In this paper, a novel proximity and load-aware resource allocation for vehicle-to-vehicle (V2V) communication is proposed. The proposed approach exploits the spatio-temporal traffic patterns, in terms of load and vehicles' physical proximity, to minimize the total network cost which captures the tradeoffs between load (i.e., service delay) and successful transmissions while satisfying vehicles's quality-of-service (QoS) requirements. To solve the optimization problem under slowly varying channel information, it is decoupled the problem into two interrelated subproblems. First, a dynamic clustering mechanism is proposed to group vehicles in zones based on their traffic patterns and proximity information. Second, a matching game is proposed to allocate resources for each V2V pair within each zone. The problem is cast as many-to-one matching game in which V2V pairs and resource blocks (RBs) rank one another in order to minimize their service delay. The proposed game is shown to belong to the class of matching games with \emph{externalities}. To solve this game, a distributed algorithm is proposed using which V2V pairs and RBs interact to reach a stable matching. Simulation results for a Manhattan model shown that the proposed scheme yields a higher percentage of V2V pairs satisfying QoS as well as significant gain in terms of the signal-to-interference-plus-noise ratio (SINR) as compared to a state-of-art resource allocation baseline.
\end{abstract}

\IEEEpeerreviewmaketitle
\section{Introduction}
\label{sec:intro}
Recently, vehicle-to-vehicle (V2V) communication has attracted significant attention as an enabler for intelligent transportation systems (ITS) \cite{GA2013}. Such applications typically require efficient proximity-aware cooperation between vehicles. These trends pose significant challenges to the underlying communication system, as vehicles must now reliably transmit their information and coordinate under latency constraints \cite{METIS}. For instance, the EU project METIS outlined that a maximum end-to-end latency of $5$~ms along with a transmission reliability of $99.999\%$ of $1600$ bytes packets should be guaranteed to deliver for such traffic safety applications \cite{METIS}.

Typically non-line-of-sight (NLOS) and line-of-sight (LOS) scenarios are studied in vehicular communications. In an urban or Manhattan model, the number of obstacles such as buildings and trees are considered as NLOS, as opposed to the freeway traffic model in which the environment is considered as LOS. Current legacy solutions for V2V communications rely on ad-hoc communications over the IEEE 802.11p standard \cite{AV2012}. Due to the dynamic nature of the vehicular communications environment and the stringent quality of service (QoS) requirements such as reliability and latency, these legacy solutions are inadequate for emerging V2V services \cite{AV2012, WX2014}. Therefore, there is a strong desire for finding better wireless networking solutions to address the challenges in V2V communications.

To address these challenges, a number of recent works have emerged focusing on V2V resource allocation \cite{BBai2011, RZhang2013, FChiti2015, MBos2014}. In \cite{BBai2011}, a low-complexity outage-optimal distributed channel allocation is proposed for V2V communications based on bipartite matching. In \cite{RZhang2013}, the authors proposed interference graph-based resource sharing schemes for resource allocation in order to enhance the network throughput. The work in \cite{FChiti2015} proposes a context-aware clustering mechanism in which coalitional game theory is used to optimize cooperation among vehicles. The authors in \cite{MBos2014} proposed a heuristic location depended resource allocation approach in which fixed number of RBs are assigned to physically disjoint zones. Furthermore, a fixed set of RBs is allocated to each zone where an RB is exclusively reused by a single vehicle user equipment (V-UE). Most of the these works on V2V resource allocation require continuous channel information for assignment of the resources to vehicles, limited by the densities of vehicles in order to reduce interference and not address the resource allocation in V2V \cite{BBai2011}, \cite{RZhang2013} and \cite{FChiti2015}, respectively.

The main contribution of this paper is to develop distributed self-organizing resource allocation mechanism for V2V communication with strict reliability requirements. The proposed approach ensures continuous (or periodic) transmission opportunities for vehicular applications such as autonomous safety services in the V2V underlay, while reducing control overhead and interference to other vehicles. Due to the localized nature of these services, we proposed a spatial and temporal resource reuse. Unlike previous works such as \cite{WX2014} and \cite{MBos2014}, our proposed resource allocation scheme takes into account the dynamic nature of vehicles' traffic patterns and their geographical information, instead of a fixed resource allocation to each location which does not changed over time \cite{MBos2014}. Due to localized nature of traffic safety applications, a new dynamic zone formation approach is proposed that allows vehicles to dynamically optimize their resource allocation, depending on the current channel quality information (CQI), reliability requirements and network topology. In such V2V networks, performing dynamic approaches require the knowledge of the entire network which incur significant overhead. Therefore, the proposed schemes allows grouping vehicle pairs into dynamic zones within which each such pair can coordinate their transmission. Moreover, we propose a novel intra-zone coordination mechanism in order to allocate resources to each vehicle pair while minimizing a cost function which captures successful transmissions and traffic load. To this end, we cast the problem as many-to-one matching game per zone with externalities in which the resources blocks (RBs) and V-UEs are the players, which rank one another based on set of preferences seeking suitable and stable allocation. To solve this game, we propose a distributed algorithm that allows RBs and V-UEs to self-organize and to maximize their own utilities within their respective zones. In addition, the proposed algorithm is shown to converge to a stable matching in which no player has an incentive to match to other player, even in the presence of externalities. Simulation results validate the effectiveness of the proposed approach and show significant performance gains compared to the baseline state-of-art approach.

The rest of this paper is organized as follows: In Section II, we present the system model and problem formulation. Section III introduces the proposed scheme for dynamic zone formation and intra-zone coordination. In Section IV, the proposed resource allocation problem is posed as matching game. Simulation results are presented in Section V. Finally, conclusions are drawn in Section VI.
\section{System Model and Problem Formulation}
\label{sec:2}
\subsection{System Model}
We consider a Manhattan grid layout in which $K$ V-UE pairs (counted in terms of transmitters) share resources. Let $\mathcal{K} = \{1, \ldots K\}$ and $\mathcal{N} = \{1, \ldots N\}$ be the set of V-UE pairs and the set of orthogonal resource blocks (RBs). We assume that all V-UEs are under the coverage of a single roadside unit (RSU). Each V-UE pair $k \in \mathcal{K}$ uses one RB whereas one RB may be shared among multiple V-UEs and cause interference. The considered scenario is illustrated in Fig \ref{fig:zones}. We further assume that each V-UE pair has a QoS requirement based on its individual packet size $1/\mu_k(t)$ and packet arrival rate $\lambda_k(t)$. In this respect, the traffic influx rate $\phi_{k}(t) = \lambda_{k}(t)/\mu_{k}(t)$ and data rate $R_{n,k}(t)$ of $k$ V-UE pair over RB $n \in \mathcal{N}$ at time slot $t$. The \emph{time load} over RB $n$ is the fractional time in slot $t$ required to deliver the requested traffic for V-UE $k$ which is defined as:
\begin{equation}
 \label{eq:time-load}
     \rho_{n,k}(t) := \frac{\phi_{k}(t)}{R_{n,k}(t)}.
\end{equation}
\begin{figure}[t]
\centering
\includegraphics[width=0.98\linewidth]{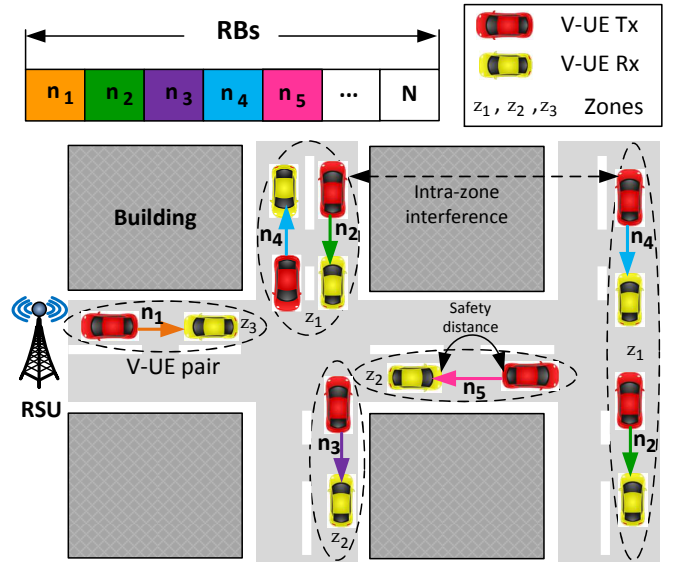}
\caption{V2V communications scenario for a Manhattan grid model with zone formation $\mathcal{Z} = \{z_1, z_2,z_3\}$. If two V-UE pairs in close proximity use the same RB, interference can be very high.}
\label{fig:zones}
\vspace{-0.5cm}
\end{figure}
Note that, for any V2V pair, a successful transmission implies that the transmitter of a V-UE pair delivers the traffic to its respective receiver, i.e., $\rho_{n,k}(t) \in (0,1)$ for any V-UE pair $k \in \mathcal{K}$. Since fast moving vehicles induce shorter channel coherence times, frequent gathering CQI is not feasible thus, a CQI that periodically varies over a time period $T$ is considered at the RSU level. Therefore, the expected time load estimation $\bar{\rho}_{k}$ for V-UE $k$ is given by:
\begin{align}
\label{eq:exp-time-load}
     \textstyle \rho_{k}(t) & =\textstyle \frac{1}{|\mathcal{N}|} \sum\limits_{n=1}^{N} \rho_{n,k}(t).
\end{align}
\begin{align}
\label{eq:exp-time-load-1}
	\textstyle \bar{\rho}_k \triangleq \frac{1}{T} \sum\limits_{t=1}^{T} \rho_{k}(t).
\end{align}
The achievable data rate for V-UE pair $k \in \mathcal{K}$ over RB $n$ at location $\boldsymbol{x}$ is written as:
\begin{equation}
      \label{eq:rate}
	        R_{n,k}(t) = \omega_{n} \log_2 \big( 1 + \gamma_{n,k}(t) \big),
\end{equation}
where $\omega_n$ is the bandwidth of RB $n \in \mathcal{N}$. The achievable SINR $\gamma_{n,k}(t)$ at V-UE pair $k \in \mathcal{K}$ over RB $n$ at location $\boldsymbol{x}$ at time $t$ is:
\begin{equation}
    \label{eq:sinr}
		 \gamma_{n,k}(t)= \frac{P_{n,k} h_{n,k}(t) } {\sum_{k' \in \mathcal{K} \setminus k} \; P_{n,k'} h_{n,k'}(t) + N_0},
\end{equation}
where $P_{n,k}$ is the transmission power and $h_{n,k}(t)$ is the channel gain of V-UE pair $k$ over RB $n$, while $N_0$ is the noise spectral density. The interference term in the denominator represents the aggregate interference at V-UE $k$ caused by the transmissions of other V-UEs $k' \in \mathcal{K} \setminus \{k\}$ on same RB. Here, for tractability, we assume that RB $n$ allows V-UE $k$ to transmit if the experienced $\gamma_{n,k}(t)$ exceeds a target SINR threshold $\overline{\gamma}_{n,k}(t)$. Otherwise, V-UE $k$ cannot transmit on RB $n$, such that $\overline{\gamma}_{n,k}(t)$ is equal for each V-UE pair. Due to the fact that vehicles are continuously changing their location and generating different traffic load over time, it is not practical to perform resource allocation at each time instance.
\subsection{Problem Formulation}
Using a central controller to allocate resources for every V2V pair is impractical due to fast-varying CQI and stringent reliability constraints. Our objective is to propose an efficient and self-organizing resource allocation approach while eliminating the frequent collection of CQI information. In this section, we study the problem of V-UEs zone formation and flexible resource allocation by incorporating vehicles' geographical information and traffic patterns. Then, we use the framework of matching theory \cite{YGu2015}, to develop a distributed and self-organizing solution. Whenever an RB is allocated to V-UE pair, it should maintain QoS requirement via a resource allocation process. For proper resource allocation V-UE pairs needs to coordinate with the rest of the network through a central controller and thus, incurs large information exchange among vehicles. Therefore, a number of V-UE pairs are grouped into sets of zones\footnote{The terms zone and cluster are used interchangeably.} based on mutual interference and traffic load. In view of above, we let $\mathcal{Z}$ be the set of zones. Each zone $z \in \mathcal{Z}$ is of dynamic size and changes over time according to proximity-based information and traffic load. Let, $\eta_z$ be the per-zone resource allocation variable which represents the mapping (matching) of V-UE pair $k$ and RB $n$ within zone $z$. Each zone $z$ consists of temporally and spatially coupled V-UEs due to mutual interference. Let $S_z$ be the total number of V-UE pairs satisfying the target SINR inside zone $z \in \mathcal{Z}$. Here, we let $\mathcal{K}_z$ denote the set of V-UEs belonging to zone $z$ and $\mathcal{N}_z$ denote the set of orthogonal resource blocks assigned to zone $z$. In each zone $z$, V-UEs efficiently reuse resources while simultaneously satisfying their QoS. Moreover, the allocation vector for each zone are collected in a vector $\boldsymbol{\eta} = [\eta_{z_1}, \ldots  \eta_{|\mathcal{Z}|}]$. Hereinafter, we refer to $\boldsymbol{\eta}$ as the ``network wide matching'' of all zones. The total expected time load of each zone is calculated as the aggregated load of each zone member $\bar{\rho}_z(\eta_z)= \sum_{k \in \mathcal{K}_z}\bar{\rho}_{k}$. We define the per-zone cost $\Gamma_z(\eta_z)$ for a given mapping $\eta_z$ by:
\begin{equation}
    \label{eq:zone-cost}
        \Gamma_z (\eta_z) =  \frac{[\bar{\rho}_z (\eta_z)]^\alpha}{[\frac{S_z}{K_z}]^\beta},
\end{equation}
where the coefficients $\alpha$ and $\beta > \alpha$ are weight parameters that indicate the impact of the load and number of satisfied V-UE pairs. Our objective is to minimize the network cost, which is given by the following optimization problem:
\begin{subequations}
        \label{eq:mainprb}
            \label{eq:main-prb}
            \begin{eqnarray}
            \underset{\{ \boldsymbol{\eta}, \mathcal{Z} \} }{\text{minimize}} && \sum_{z \in \mathcal{Z}} \Gamma_z (\eta_z), \\
            \text{subject to}
            && |\mathcal{K}_z| \geq 1,\; \forall z \in \mathcal{Z} \label{eq:main-prb-2} \\
            &&  \mathcal{K}_{z} \cap \mathcal{K}_{z'} = \emptyset, \; \forall z,z' \in \mathcal{Z}, \; z \neq z', \label{eq:main-prb-3} \\
            && \sum_{n=1}^{N} \eta_z(k,n) \leq 1, \; \forall k \in \mathcal{K}_z, \label{eq:main-prb-4} \\
            && \bigcup_{\forall k \in \mathcal{K}_z} \eta_z(k,n) \leq |\mathcal{N}_z|, \; \forall z, \forall n. \label{eq:main-prb-5}
            \end{eqnarray}
\end{subequations}
Here, constraints \eqref{eq:main-prb-2} and \eqref{eq:main-prb-3} ensure that any V-UE is part of one zone only. \eqref{eq:main-prb-4} captures the fact that a given V-UE pair $k$ can be matched to only one RB $n$ while an, RB $n$ can be matched to one or more V-UEs for a given matching $\eta_z$. Constraint \eqref{eq:main-prb-5} states that V-UE pairs can reuse the resources within set $\mathcal{N}_z$.
\section{Dynamic Zone Formation}
The centralized optimization problem in \eqref{eq:main-prb} is challenging to solve as it is combinatorial in nature. To that end, a decentralized approach based on minimal coordination between neighboring V-UEs is proposed. First, a dynamic zone formation mechanism is devised which incorporates, both geographical proximity of V-UEs and their traffic nature (e.g., traffic load, interference). Let $G = (\mathcal{K},\mathcal{E})$ be an undirected graph, where $\mathcal{K}$ is the set of V-UE pairs and $\mathcal{E} \subset \mathcal{K} \times \mathcal{K}$ is the set of links between locally-coupled V-UE pairs in terms of distance and traffic load.
\vspace{-0.1cm}
\subsection{Neighborhood based Gaussian similarity }
Given the graph $G = (\mathcal{K},\mathcal{E})$, let $\boldsymbol{v}_k$ and $\boldsymbol{v}_{k'}$ be the geographical coordinates of the V-UE pairs $k$ and $k'$ respectively. The geographical locations of vehicles do not change significantly during time $1 \leq t \leq T$. Therefore, average locations of vehicles are considered. Here, we use parameter $\epsilon_d$ to represent the presence of a link or edge $e \in \mathcal{E}$ between V-UE pairs $k$ and $k'$. The Gaussian distance similarity is based on the distance between two V-UEs $k$ and $k'$ as:
\begin{equation}
    \label{eq:distance-clustering}
	d_{k,k'} = \begin{cases}	\exp \bigg( \frac{-||\boldsymbol{v}_{k} - \boldsymbol{v}_{k'}||^2}{2\sigma_d^2} \bigg), & \mbox{if $ \lVert \boldsymbol{v}_{k} - \boldsymbol{v}_{k'} \rVert \leq \epsilon_d$},\\
              0, & \mbox{otherwise}.
               \end{cases}
\end{equation}
Here, the parameter $\sigma_d$ controls the impact of the neighborhood size. For a given $\epsilon_d$ is the range of the Gaussian distance similarity for any two close-by V-UE pairs is $[e^{-\epsilon_d/2\sigma_d^2},1]$, where the lower bound is determined by $\sigma_d$. Furthermore, $\boldsymbol{D}$ represents the distance based similarity matrix with each entry $d_{k,k'}$ given by \eqref{eq:distance-clustering}.
\vspace{-0.1cm}
\subsection{Load based similarity}
Comparing the variation in geographical location of V-UEs pairs in (\ref{eq:distance-clustering}), the traffic patterns of V-UEs varies more frequent over time. Therefore, grouping V-UEs pairs based on their temporal traffic aspects provides more dynamic manner of zone formation. To provide a notion of how similar two V-UE pairs are in terms of expected traffic load over time, we use the cosine similarity matrix $\boldsymbol{C}$.  Here, we consider a periodic time interval $T$ during which the expected time load $\rho_k(t)$ of each V-UE pair is observed. We build a expected traffic load vector $\boldsymbol{\rho_k} = [ \rho_k(1), \ldots \rho_k(T) ]$ in order to calculate the cosine similarity between two V-UEs $k$ and $k'$ as:
\begin{equation}
    \label{eq:load-clustering}
		c(k,k') = \frac{\boldsymbol{\rho_k} \cdot \boldsymbol{\rho_{k'}}} {\lVert \boldsymbol{\rho_k} \rVert \lVert \boldsymbol{\rho_{k'}} \rVert},
\end{equation}
The value of $c(k,k')$ is used to as the $(k,k')$-th entry of the load-based similarity matrix $\boldsymbol{C}$.
\vspace{-0.1cm}
\subsection{Combining the similarities}
The key step in zone formation is to identify similarities between V-UEs pairs to group V-UEs with similar characteristics. In order to reuse resource within a given zone grouping V-UEs pairs having different traffic and physically apart from each other, we combine load and distance similarity. An affinity matrix $\boldsymbol{A}$ that blends time-average load affinity with spatial proximity is:
    \begin{equation}
        \label{eq:affinity-matrix}
        \boldsymbol{A} = \theta(1 - \boldsymbol{C}) + (1-\theta)(1 - \boldsymbol{D}),
    \end{equation}
where $\theta$ control the impact of load similarity and distance, respectively. Algorithm \ref{algo:algo1} describes the proposed V-UE zone clustering method. Grouping dissimilar V-UEs in terms of traffic and geographical distance, mitigates interference and thus, minimizes the total network cost. Moreover, the Hare-Niemeyer method is used for calculating the set of $\mathcal{N}_z$ for each zone \cite{HeraQuota}. However, solving \eqref{eq:main-prb} requires knowledge of all zones in the network, which can be complex and not practical. Therefore, next, section we propose a distributed self-organizing solution for flexible V-UE-RB allocation based on intra-zone coordination.
\begin{algorithm}[t]
\footnotesize
\caption{Spectral clustering for zone formation}
\label{algo:algo1}
\DontPrintSemicolon 
\textbf{Initialization}: Pick a time window of length $T$, calculate affinity matrix $\boldsymbol{A} = [a(i,j)]$ as in \eqref{eq:affinity-matrix} of a graph $G$, choose $b_{\text{min}}= 2$ and $b_{\text{max}} = K/2$\;
Compute diagonal degree matrix $\boldsymbol{M}$ with diagonal $m_i = \sum_{j=1}^{K} a_{i,j}$. \;
$\boldsymbol{L} = \boldsymbol{M} - \boldsymbol{A}$ \;
$\boldsymbol{L}_{\textrm{norm}} := \boldsymbol{M} ^{-1/2} \boldsymbol{L} \boldsymbol{M} ^ {-1/2}$. \;
Pick a number of $b_{\text{max}}$ eigenvalues of $\boldsymbol{L}_{norm}$ such that $\lambda_1 \leq, \ldots, \leq \lambda_{k_{\text{max}}}$. \;
Choose $B = {\text{max}}_{i=b_{\text{min}}, \ldots, b_{\text{max}}} \Delta_i$ where $\Delta_i = \lambda_{i+1}- \lambda_i $. \;
Choose the $b$ smallest eigenvectors $x_1, ..., x_b$ of $\boldsymbol{L}_{\textrm{norm}}$. \;
Let $\boldsymbol{Y}$ matrix has the eigenvectors $x_1, \ldots x_b$ as columns. \;
Use k-means clustering to cluster (zone) the rows of matrix $\boldsymbol{Y}$. \;
Zone set $\{1, ..., Z_{|\mathcal{Z}|}\}$. \;
\end{algorithm}
\vspace{-0.1cm}
\section{Intra-zone resource Allocation as Matching game with Externalities}
\label{sec:3}
Our objective is to develop a self-organizing mechanism to solve the resource allocation problem in \eqref{eq:main-prb} using a decentralized, approach in which V-UEs and RSU-controlled RBs interact and make resource allocation decision based on their utilities. To this end, the framework of matching theory \cite{YGu2015} is a promising approach for resource management in V2V communication. Prior to defining the game, we will formulate the load-aware utility functions for both players (V-UEs and RBs) to optimize the resource allocation mechanism.
\subsection{Utilities of V-UEs and RBs}
For RB allocation, each V-UE needs to identify the set of RBs that ensure QoS requirements and reduce their load. Thus, we define a suitable load-aware utility function for any V-UE $k \in \mathcal{K}_z$ over RB $n \in \mathcal{N}_z$ for a given matching $\eta_z$ as follows:
\begin{align}
\label{eq:utility-vue}
       U_{n,k}(\eta_z) & = -{\rho}_{n,k}(\eta_z).
\end{align}
Consequently, the utility of an RB is the sum of utilities of V-UE pairs using that RB for the given matching $\eta_z$:
\begin{align}
\label{eq:utility-rb}
     U_{n}(\eta_z)  & = \sum_{k \in \mathcal{K}_z} U_{n,k}(\eta_z).
\end{align}
\indent
Finally, we define a zone-wide utility for zone $z$ as per the given matching $\eta_z$ as $ W_z(\eta_z) = -\Gamma_z(\eta_z) $. Having defined such utilities, our goal is to maximize the zone-wide utility to solve \eqref{eq:main-prb} in which each V-UE $k \in \mathcal{K}_z$ is assigned to an RB $n \in \mathcal{N}_z$ via a matching $\eta_z: \mathcal{K}_z \rightarrow \mathcal{N}_z$.
\begin{defn}
\label{def:1}
A \emph{matching game} is defined by two sets of players ($\mathcal{K}_z, \mathcal{N}_z$) and two preference relations $\succ_k$ and $\succ_n$ for each V-UE $k \in \mathcal{K}_z$ to build its preference over RB $n \in \mathcal{N}_z$ and vice-versa in a zone $z$.
\end{defn}
The outcome of the matching game is the allocation mapping $\eta_z$ that matches each player $k \in \mathcal{K}_z$ to player $n=\eta_z(k)\; n\in \mathcal{N}_z$ and vice versa such that $k = \eta_z(n)\;, k \in \mathcal{K}_z$.
\setlength{\textfloatsep}{1pt}
\LinesNumberedHidden{
\begin{algorithm}[t]
\footnotesize
\caption{Dynamic Proximity-aware resource allocation algorithm}
\label{algo:algo2}
\DontPrintSemicolon 
\KwData{Each V-UE $k$ is initially allocated to a randomly selected RB $n$.}
\KwResult{Convergence to a stable matching $\boldsymbol{\eta}$.}

\While{$t \leq  T_{\text{max}}$} {

\textbf{Phase I - Load and SINR computation;}

\begin{itemize}
    \item Calculate SINR and time load over each V-UE $k$ at time $t$ using \eqref{eq:time-load} and \eqref{eq:sinr};
\end{itemize}
\If{ $t/T \in \mathbb{N}$}{
\textbf{Phase II - Zone formation between V-UE pairs;}
\begin{itemize}
    \item Calculate gaussian distance and cosine load dissimilarity metrics \\ using \eqref{eq:distance-clustering}, \eqref{eq:load-clustering};
    \item Affinity matrix computed using \eqref{eq:affinity-matrix};
    \item Zones $|\mathcal{Z}|$ are formed among V-UE pairs using Algorithm \ref{algo:algo1};
\end{itemize}

\textbf{Phase III - RB Allocation and utility computation;}
\begin{itemize}
    \item Calculate set of RBs for each zone $\mathcal{N}_z$ \cite{HeraQuota}.
	\item $U_{n,k}(\eta_z)$, $U_n(\eta_z)$ and zone-wide utility $W_z(\eta_z)$.
\end{itemize}

\textbf{Phase IV - Swap-matching evaluation;}

   \While{$z \leq \max(|\mathcal{Z}|)$} {
   \begin{itemize}
	   \item Pick a random pair of V-UEs $\{k,\tilde{k}\} \in \mathcal{K}_z$ within $z$ zone;
    \end{itemize}

    	\While{$count \leq  count_{\text{max}}$} {
    	\begin{itemize}
    			 \item $U_{n,k}(\eta_z)$, $U_n(\eta_z)$ are updated;
   				    \item swap the pair of UEs $\eta_z(\{k,\tilde{k}\} \{\tilde{k},k\})$
    				\item $ W_z(\eta_z) = W_z(\eta_z)_{\text{best}}$
    	\end{itemize}
    \If{$ W_z(\eta_z) > W_z(\eta_z)_{\text{best}}$ }{\
         $W_z(\eta_z)_{\text{best}} = W_z(\eta_z)$ \;
        }
    \Else{
      RB $n$ refuses the proposal, and V-UE $k$ sends a proposal to the next configuration at $count$ \;
    }
		$count = count + 1$
    }
    $z = z + 1$
}
\textbf{end if}}
$t = t + 1$
}
\textbf{Phase V - Stable matching}
\end{algorithm}}
A \emph{preference relation} $\succ$ is defined as a reflexive, complete and transitive binary relation between players in $\mathcal{K}_z$ and $\mathcal{N}_z$. Thus, a preference relation $\succ_k$ is defined for every V-UE $k \in \mathcal{K}_z$ over the set of RBs $\mathcal{N}_z$ such that for any two nodes in $n,\tilde{n} \in \mathcal{N}^{2}_z, n \neq \tilde{n} $ and two matchings $\eta_{z}, \eta'_{z} \in \mathcal{K}_z \times \mathcal{N}_z, \eta_z \neq \eta'_z \;, n =\eta_z(k)\;, \tilde{n}=\eta'_{z}(k)$:
\begin{align}
    \label{eq:vue-pref}
   & (n,\eta_z) \succ_k (\tilde{n},\eta'_z)  \Leftrightarrow U_{n,k}(\eta_z) > U_{\tilde{n},k}(\eta'_z),
\end{align}
where $(n,k) \in \eta_z$ and $(\tilde{n},k) \in \eta'_z$. Similarly the preference relation $\succ_n$ for RB $n$ over the set of V-UEs $\mathcal{K}_z$ is defined such that for any two V-UEs $k,\tilde{k} \in \mathcal{K}_z, k \neq \tilde{k}\;, k =\eta_z(n)\;, \tilde{k}=\eta'_{z}(n)$:
\begin{equation}
 \label{eq:rb-pref}
    (k,\eta_z) \succ_n (\tilde{k}, \eta'_z) \Leftrightarrow U_{n}(\eta_z) > U_{\tilde{n}}(\eta'_z).
\end{equation}
Each V-UE and RB independently rank one another based on the utilities in \eqref{eq:utility-vue} and \eqref{eq:utility-rb}. However the selection preferences of V-UE \eqref{eq:vue-pref} and RB \eqref{eq:rb-pref} are \emph{interdependent} and influenced by the existing network wide matching, which leads to a many-to-one matching game. Such effects which dynamically change the preference of each player in the network, are called externalities. In particular, the considered game is a matching game with \emph{externalities} due to mutual interference between vehicles, which differs from classical applications of matching theory in wireless such as those in \cite{EA2011} and \cite{SBayat2012}. Indeed, these existing approaches do not account for externalities, and, thus, they may yield lower utilities or may not converge.
\vspace{-0.2cm}
\subsection{Proposed resource allocation algorithm}
\label{sec:51}
The concept of externalities requires us to adopt a suitable stability concept based on the idea of ``pairwise stability'' \cite{Eliz2011}. Before defining the pairwise stability, we first define a swap matching.
\begin{defn}
\label{def:2}
A \emph{swap matching} is formally defined as $\mathord{\eta_z^{k \leftrightarrow \tilde{k}} = \{ \eta_z \setminus \{(n,k),(\tilde{n},\tilde{k}) \}\} \cup \{ (k,\tilde{n}),(\tilde{k},n)\}}$. In each swap two V-UEs change their matching with their respective RBs while other matchings remain fixed.
\end{defn}
Having defined swap matching, we further define \emph{pairwise stability}. Given a matching $\eta_z$, pairs of V-UEs $k,\tilde{k}$ and RBs $n, \tilde{n}$ within a zone $z$, a pairwise matching is \emph{stable} if and only if there does not exist  pairs of V-UEs $(k,\tilde{k})$ such that:
\begin{enumerate}[(i)]
      \item $\mathord{\forall y \in \{k,\tilde{k}, n,\tilde{n} \}, \; U_{y,\eta_z^{k \leftrightarrow \tilde{k}}(y)}(\eta_z) \geq U_{y, \eta_z(y)}(\eta_z)}$
      \item $\mathord{\exists y \in \{k,\tilde{k},n,\tilde{n}\},\; U_{y,\eta_z^{k \leftrightarrow \tilde{k}}(y)}(\eta_z) > U_{y, \eta_z(y)}(\eta_z)}$
\end{enumerate}
yields $W_z(\eta_z^{k \leftrightarrow \tilde{k}}) > W_z(\eta_z)$.

A matching $\eta_z$ is said to be pairwise stable if there does not exist any V-UE $\tilde{k}$ or RB $\tilde{n}$, for which RB $n$ prefers V-UE $\tilde{k}$ over V-UE $k$ or any V-UE $k$ which prefers RB $\tilde{n}$ over $n$. From Definition \ref{def:2}, we can see that if two V-UEs swap between two RBs, the RBs involved in the swap must ``approve'' the swap. Similarly, if two RBs want to swap between two V-UEs, the V-UEs and RBs must agree to the swap. Such pairwise matching stability is reached by guaranteeing only stable swap-matching among players \cite{Eliz2011} which result in an increase in the zone-wide utility. In fact, due to externalities, players continuously change their preference orders, in response to the formation of other V-UE-RB which renders classical deferred acceptance solutions such as in \cite{EA2011, YWu2011} not applicable for our model. Therefore, to seek a stable resource allocation an Algorithm \ref{algo:algo2} is proposed, which depends on the swap which results in an increase of zone-wide utility.
\begin{table}[t]
\caption{Simulation Parameters}
\centering
\begin{tabular}{| p{5.5cm} | p{2.5cm} |}
\hline
\textbf{Parameter} & \textbf{Value} \\ \hline
Carrier frequency, RB bandwidth                         & $800$ MHz, $180$ KHz \\ \hline
Building breadth                                        & $100$~m \\ \hline
Noise power spectral density ($N_0$)                    & -$174$ dBm/Hz \\ \hline
Minimum safety distance                                 & [$15$~m $- 20$~m] \\ \hline
Range of neighborhood  ($\epsilon_d$)                   & $100$~m \\ \hline
Impact of neighborhood width ($\sigma_d$)               & $100$ \\ \hline
Impact of similarities ($\theta$)                       & $0.3$ \\ \hline
Impact of load and satisfied V-UEs ($\alpha, \beta$)    & $1$, $3$ \\ \hline
Vehicle speed                                           & $50$ km/h \\ \hline
V-UE transmission power                                 & $10$ dBm \\ \hline
V-UE target SINR $\overline{\gamma}_{n,k}(t)$           & $3$ dB \\ \hline
Channel model                                           & Rayleigh fading \\ \hline
Mean packet size $1/\mu_k(t)$              & $1600$ bytes \cite{METIS}\\ \hline
\end{tabular}
\label{tab:simulation}
\end{table}
In the first phase, each V-UE is initially allocated to a randomly selected RB. Next, the SINR is calculated at each V-UE, in order to calculate the time-average load. In the second phase, zones are formed among V-UE pairs based on the gaussian distance and cosine load similarities using Algorithm \ref{algo:algo1}. Allocation of RBs inside zone is performed based on the expected value of load at each zone. Then, the utilities of all players and utility of the zone is calculated for the given matching $\eta_z$. In the forth phase, V-UE and RBs update their respective utilities and individual preferences over one another. Subsequently, at each iteration, a chosen V-UE is swapped with other V-UE that depends on the increases in zone-wide utility.
\begin{lem}
Upon convergence of Phase IV, Algorithm \ref{algo:algo2} reaches a stable matching.
\end{lem}
\begin{proof}
The proof follows based on following considerations. First, due to limited number of V-UEs and RBs inside zone, the number of possible swaps is finite. Moreover, only the swaps which strictly improves the zone-wide utility strictly can permitted. In addition, considering all the possible swaps each V-UE is associated to its most preferred RB and vice versa. Therefore, Phase IV, terminates, when no further improvement in zone-wide utility is achieved by all possible swaps among players. Therefore, Algorithm \ref{algo:algo2} converges to a stable matching after a finite number of iterations.
\end{proof}

\begin{figure}[t]
\centering
\includegraphics[width=0.95\linewidth]{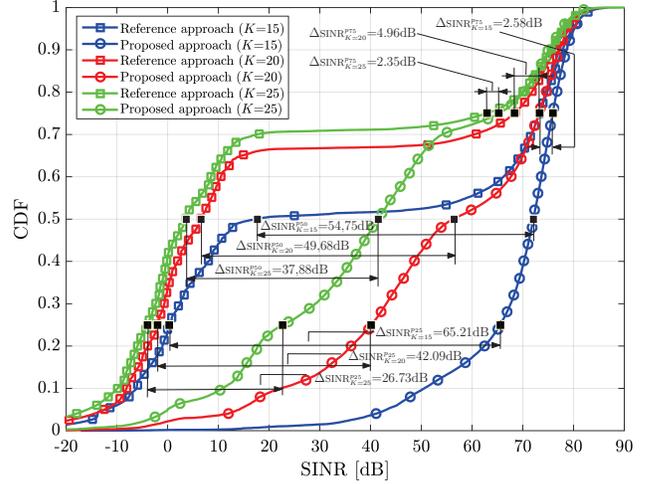}
\caption{Cumulative density function of V-UE's SINR for different densities of V-UE pairs under the considered approaches with $\Delta${\small SINR$_{K=x}^{\boldsymbol{p}_{y}}$} being the SINR gains offset for $K$ pairs in percentile $\boldsymbol{p}_y=[25,50,75]$.}
\label{fig:SINR-cdf}
\end{figure}
\vspace{-0.1cm}
\section{Simulation Results}
\label{sec:4}
For simulations, we consider a single cell in which a number of V-UE pairs are uniformly distributed over a Manhattan grid model. The radio resources are organized in $15$ RBs dedicated for V2V communication. Four building are deployed, each of them with the breadth of $100$~m, spaced by two lanes for bi-directional traffic mobility. We consider four types of vehicles with different length and widths. Due to presence of the building each vehicle encounter LOS and NLOS instances. The path loss model for V2V communication of LOS and NLOS is computed according to the Berg recursive model \cite{METIS}. V-UEs are traveling at a fixed speed of $50$ km/h along the defined roads. The positions of the V-UEs in the cell are assumed to be exactly known. The simulation parameters are given in Table \ref{tab:simulation}. In order to compare the proposed approach with the idea proposed in state-of-art \cite{MBos2014}, total area is divided into fixed equal size zones such that fixed number of RBs are allocated to each zone while zones having more load in terms of traffic gets more RBs. The performance of the proposed approach is analyzed under the different load conditions. Zones are formed after each $10$ time slots, and simulation is run for $T_{\textrm{max}}=60$ seconds.

Fig. \ref{fig:SINR-cdf} shows the cumulative density function of the V-UEs' SINR for a fixed number of RBs $N=15$ and varying density of V-UE $K$ pairs. From Fig. \ref{fig:SINR-cdf}, we can see that, the proposed approach significantly improves the SINR. Indeed, Fig. \ref{fig:SINR-cdf} shows that there are less then $6.2\%$ V-UEs when $K=25$ that do not satisfy the target SINR due to NLOS, but overall in the proposed approach significantly outperforms the reference approach in $25$-th, $50$-th and $75$-th percentiles of V-UE SINR.\\
\begin{figure}[t]
\centering
\includegraphics[width=0.90\linewidth]{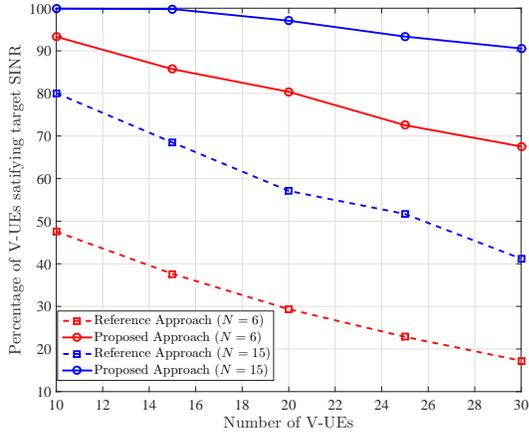}
\caption{Percentage of V-UEs salifying the target SINR for different densities of V-UEs and different number of RBs under the considered approaches.}
\label{fig:per-satisfied}
\end{figure}
In Fig. \ref{fig:per-satisfied}, we show the percentage of V-UEs which achieved the desired target SINR for different number of V-UEs and RBs. From this figure, we can see that the proposed approach significantly reduces the number of outages compared to the reference scheme for all network sizes. Fig. \ref{fig:per-satisfied} shows that, almost $100\%$ V-UEs pairs satisfying the target SINR in the case of $10$ and $15$ V-UE pairs when $N=15$ RBs. The advantage of the proposed approach reaches up to $49\%$ and $50\%$ of gain when compared to the baseline approach with $K=30$ V-UEs for $N=6$ and $N=15$ RBs, respectively.

Fig. \ref{fig:iterations} shows the number of swap iterations required for the convergence of algorithm for the same scheduling time instance. Due to dynamic nature of the zone formation, number of zones varies over time. Each zone tries to maximize its utility locally. In this figure, we can see that for lower number of V-UE pairs, fewer number of iterations are required. As increases in the number of V-UE pairs inside zone, it requires more iterations for convergence. From Fig. \ref{fig:iterations} we can also observe that, the proposed approach requires reasonable number of iterations for convergence.
\begin{figure}[t]
\centering
\includegraphics[width=0.90\linewidth]{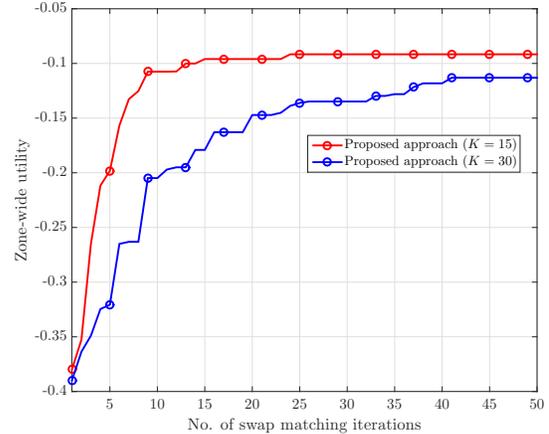}
\caption{Number of swap matching iterations per zone for different densities of V-UEs at the same scheduling time instance}
\label{fig:iterations}
\end{figure}
\vspace{-0.15cm}
\section{Conclusions}
\label{sec:5}
In this paper, we have presented a novel, a dynamic proximity-aware resource allocation scheme for V2V communications. Above all, the aim of proposed work is to satisfy the requirements of V2V safety services while reducing the signaling overhead and interference from other pairs by enabling zone formation. We have formulated the problem as a matching game with externalities in which the V-UEs and RBs build preferences over one another so as to choose their own utilities. Simulation results have shown that the proposed approach provides considerable gains in terms of increased SINR and higher percentage of V-UEs satisfying QoS with respect to a state-of-the-art resource allocation baseline.
\section*{Acknowledgment}
This research was supported  by TEKES grant 2364/31/2014 and the Academy of Finland (CARMA) and the U.S. National Science Foundation under grants CNS-1460316 and ACI-1541105. We also acknowledge Sumudu Samarakoo and Petri Luoto for their fruitful discussions.

\vspace{-0.1cm}
\renewcommand{\baselinestretch}{0.85}

\end{document}